\documentclass[conference]{IEEEtran}

\usepackage{latexsym}
\usepackage{amssymb}
\usepackage{amsmath}
\usepackage{graphics}
\usepackage{graphicx}
\usepackage{color}

\ifCLASSINFOpdf
\else
\fi

\begin{document}
\title{Control-Induced Decoherence-Free Manifolds}
\author{\IEEEauthorblockN{E. Jonckheere}
\IEEEauthorblockA{Department of Electrical Engineering\\
University of Southern California\\
Los Angeles, CA 90089\\
Email: jonckhee@usc.edu}
\and
\IEEEauthorblockN{A. Shabani}
\IEEEauthorblockA{Department of Chemistry\\ 
Princeton University\\
Princeton, NJ 08544\\
Email: ashabani@princeton.edu}
\and
\IEEEauthorblockN{A. T. Rezakhani}
\IEEEauthorblockA{Department of Chemistry\\
University of Southern California\\
Los Angeles, CA 90089\\
Email: tayefehr@usc.edu}}

\maketitle
\newtheorem{theorem}{Theorem}
\newtheorem{lemma}{Lemma}
\newtheorem{corollary}{Corollary}
\newtheorem{conjecture}{Conjecture}
\newtheorem{problem}{Problem}
\newcommand{\R}{\mathcal{R}}
\newcommand{\V}{\mathcal{V}}
\renewcommand{\S}{\mathcal{S}}
\newcommand{\G}{\mathcal{G}}
\newcommand{\N}{\mathcal{N}}
\newcommand{\B}{\mathcal{B}}
\renewcommand{\i}{\imath}
\newcommand{\W}{\mathcal{W}}
\newcommand{\M}{\mathcal{M}}
\renewcommand{\L}{\mathcal{L}}
\newcommand{\A}{\mathcal{A}}
\renewcommand{\H}{\mathcal{H}}
\newcommand{\Herm}{\mathrm{Herm}}
\newcommand{\D}{{\mathcal D}}
\newcommand{\im}{\mbox{im}}
\newcommand{\E}{\mathcal{E}}
\renewcommand{\>}{\rangle}
\newcommand{\<}{\langle}
\newcommand{\x}{\bar{x}}
\newcommand{\n}{\bar{n}}
\renewcommand{\k}{\bar{k}}
\newcommand{\K}{\bar{K}}

\begin{abstract}

Quantum coherence of open quantum systems is usually compromised because of the interaction with the ambient environment. A ``decoherence-free subspace'' (DFS)  of the system Hilbert space is defined where the evolution remains unitary. In the absence of \textit{a priori} existence of such subspaces, it seems natural that utilizing quantum control may help generate and/or retain a DFS. Here, we introduce a time-varying DFS wherein the system's density matrix has a unitarily evolving sub-density corresponding to some given set of its eigenvalues (which we aim to preserve). This subspace is characterized from both topological and algebraic perspectives. In particular, we show that this DFS admits a complex vector bundle structure over a real-analytic manifold (the decoherence-free manifold).
\end{abstract}
\IEEEpeerreviewmaketitle

\section{Introduction}

\subsection{Motivation}

Eliminating environmental effects, or decoherence, in quantum systems is a main challenge of engineering devices that make direct use of
quantum mechanical rules for information processing. In past few years, various
methods have been proposed and implemented to mitigate the deleterious effect of
decoherence in quantum computers and communication systems. For example, the notion of \emph{decoherence-free spaces} was introduced as
a passive method to bypass decoherence \cite{x1,x2,DFS_from_Alireza}. In these methods,
information is stored and processed in a protected subspace of the system Hilbert space,
a subsystem, or a hybrid form of them \cite{x1,x2,DFS_from_Alireza}.

States of a quantum system are represented by density matrices, trace-$1$ positive semidefinite $n\times n$ matrices $\varrho\in\mathcal{S}(\mathcal{H})$ defined on the linear space $\mathcal{S}(\mathcal{H})$ of the linear operators acting on the system's Hilbert space $\mathcal{H}\cong \mathbb{C}^n$. In the absence of decoherence, evolution of a closed quantum system is described by a unitary transformation $V(t)\in U(n)$,
\begin{eqnarray}
\varrho(t)=V(t)\varrho(0)V^*(t),
\end{eqnarray}
whereas decoherence generally results in non-unitarity. A non-unitary transformation implies
irreversibility of the dynamics, hence loss of information. However, certain \textit{symmetries} of the system dynamics can yield a unitary sub-dynamics in some part of $\mathcal{H}$. Roughly speaking, a DFS is the subspace associated with such a sub-dynamics.

In a general setting, $\mathcal{H}$ can be decomposed as
\begin{eqnarray}
& \mathcal{H}=\oplus_{j} \mathbb{C}^{n_j}\otimes \mathbb{C}^{\n_j},
\label{decomposition}
\end{eqnarray}
where each $\mathbb{C}^{n_j}\otimes \mathbb{C}^{\n_j}$ is a subspace of $\mathcal{H}$ and $\mathbb{C}^{n_j}$ ($\mathbb{C}^{\n_j}$) denotes a subsystem in this subspace \cite{x2}. A subsystem $\mathbb{C}^{n_j}$ is termed \emph{decoherence-free} if its corresponding sub-dynamics is unitary. The state of the subsystem with the Hilbert space $\mathbb{C}^{n_j}$
is found to be $\varrho_j=\mathrm{Tr}_{\n_j}[\mathcal{P}_{n_j,\n_j}\varrho \mathcal{P}_{n_j,\n_j}]$, where $\mathcal{P}_{n_j,\n_j}$is the projector on $\mathbb{C}^{n_j}\otimes \mathbb{C}^{\n_j}$ and $\mathrm{Tr}_{\n_j}[\cdot]$ is the partial trace over $\mathbb{C}^{\n_j}$. If a state $\varrho_j$ is decoherence-free, then for all times $t$ there exists a unitary $V_j(t)$ such that $\varrho_j(t)=V_j(t) \varrho_j(0) V^{*}_j(t)$.

From here on, we focus on decoherence-free subspaces (DFS), 
and represent the decoherence-free state and projector by $\varrho_{\text{DFS}}$ and $\mathcal{P}_{\text{DFS}}$, respectively. In the traditional 
definition of DFS, it is further assumed that the projector $\mathcal{P}_{\text{DFS}}$ is time-independent, meaning that
the DFS is time-invariant. In this work, however, we relax this condition in the sense that in our approach
$\mathcal{P}_{\text{DFS}}$ is considered to be time dependent. A precursor to the concept of time-varying DFSs has been
introduced in \cite{Ognyan},
wherein unitarily-correctible subsystems are interpreted as time-varying noiseless spaces for open quantum systems.

\subsection{Set-up}

Under some fairly general conditions \cite{BP-book}, the evolution of the system in its embedding environment can be described by the following Lindblad equation [in the units of $\hbar\equiv1$]: 
\begin{eqnarray}
\partial_t\varrho=-\imath[H_0+\sum_\alpha H_{u_\alpha}u_\alpha(t),\varrho] + \sum_\alpha L_\alpha(\varrho)\gamma_\alpha(t).
\label{e:liouville_von_neumann}
\end{eqnarray}
Here, the Hermitian matrix $H_0$ denotes the free-evolution Hamiltonian of the open system (including some generically small corrections to the Hamiltonian of the isolated system because of the interaction with the environment \cite{BP-book}); $\sum_\alpha H_{u_\alpha} u_\alpha(t)$ is the {\it control} Hamiltonian, with real-valued ``knobs" $u_\alpha(t)$;  
and, the Hermitian $\sum_\alpha L_\alpha (\varrho)\gamma_\alpha$ encapsulates the interaction with the environment, responsible for decoherence. Specifically, 
\begin{eqnarray}
L_\alpha(\varrho)=\frac{1}{2} \left( [F_\alpha,\varrho F_\alpha^*]+[F_\alpha\varrho,F_\alpha^*] \right) 
\label{e:lindbladian}
\end{eqnarray}
is a ``Lindbladian,'' $F_\alpha$ is a quantum jump operator,
and $\gamma_\alpha \geq 0$ is the jump rate. 
This is a generalization of the model proposed in \cite[Eqs. (3),(6)]{mabuchi} 
in order to incorporate the uncertainty in the decoherence rates $\gamma_\alpha$ \cite{random_decay_rate},~\cite[Eq. (34)]{sonia_decoherence}. 

Note that under Eq.~(\ref{e:lindbladian}) both positivity and trace of the density matrix are preserved. That is, if $\varrho(0)$ is trace-$1$ positive semidefinite matrix, so will be $\varrho(t)$ for any $t\in\mathbb{R}^+$: $\varrho(t)\geq0$ and $\mathrm{Tr}[\varrho(t)]=1$.  
Thus, in general, the system evolves over the compact set $\D$ of trace-$1$ positive semidefinite operators. 

When $\gamma = 0$ the evolution (\ref{e:liouville_von_neumann}) is unitary, hence the eigenvalues of $\varrho(t)$ remain constant while its eigenvectors evolve unitarily. In the open system case $\gamma\neq 0$, the evolution is no longer unitary, thereby both eigenvalues and eigenvectors evolve in time. It has been shown that under some circumstances, however, one can find a DFS$~\subseteq \H$, wherein the sub-dynamics evolution still remains unitary \cite{sonia_decoherence,tracking_control_lidar,DCF1,DCF2}, hence preserving {\it some} eigenvalues of the density operator. Here, we define a \emph{decoherence-free manifold} (DFM) as the manifold embedded in the space of density operators ($\mathcal{D}\subset\mathcal{S}(\mathcal{H})$) over which a given subset of the eigenvalues, along with their multiplicities, are preserved, and the corresponding eigenvectors evolve unitarily. The corresponding eigenspace then generalizes the DFS concept in that it depends explicitly on $\varrho$ and as such is formalized as a vector bundle over DFM. 

Utilizing quantum control enables a natural setting for generating and preserving DFSs in open quantum systems. For example, real time-feedback \cite{Fran} and Lyapunov control \cite{Lya} methods have been used to generate time-independent DFSs. An auxiliary system can be also advantegous for creating DFS using open-loop control techniques \cite{Tarn}.
Decoherence control can be viewed as a disturbance rejection problem, which under classical interpretation would proceed under the assumption that $\gamma$ is stochastically varying, as it happens under complex system-reservoir interaction~\cite{random_decay_rate}. From this perspective, the disturbance rejection solution is a {\it feedback} $u(t)=f(\varrho(t))$ constructed so as to make the DFM \textit{controlled-invariant}, ignoring the back-action effect of the measurement. 
This measurement effect can make applicability of the feedback concept difficult to justify. 
However, such feedback solution can be justified, for example, in nuclear magnetic resonance (NMR) ensemble control~\cite[Sec. 7.7]{book_nielsen}, where $H_0$ is the free precession of spins, $H_uu$ is the radio-frequency excitation, and $\sum_\alpha L_\alpha \gamma_\alpha$ represents the interaction with the environment together with the measurement back-action that for an ensemble of quantum systems has a deterministic Lindbladian form \cite{Kurt}. 

However, in practical situations, $\gamma$ could be known completely or at least partially, and under such circumstances the DFM has to be made \textit{self-bounded controlled-invariant} with a \textit{feedforward} control, or a combined feedforward/feedback solution involving ``partial preview'' on the disturbance $\gamma$~\cite{self_bounded_previewed_signals}. 


\section{Generalized decoherence-free concept}
\label{s:newer}

\subsection{Fundamental concepts}
\label{s:fconcepts}

The fundamental concept is to keep one or several blocks of the eigenvalues of $\varrho(t)$ constant; as we shall see, this in fact secures a decoherence-free evolution along the corresponding eigenspace. Consider the spectral decomposition of $\varrho(t)$,
\begin{eqnarray}
& \varrho(t)=\sum_{i=1}^{n}\lambda_i(t)|e_{i}(t)\rangle \langle e_i(t)|,
\end{eqnarray}
where $0\leq \lambda_i(t)\leq 1$, $\sum_{i=1}^{n}\lambda_i(t)=1$, $\lambda_i(t)\leq \lambda_j(t)$ for $i<j$, $|e_i(t)\rangle\in\mathcal{H}$, $\langle e_i(t)|\in\text{dual}(\mathcal{H})$, and $\langle e_i(t)|e_j(t)\rangle=\delta_{ij}$. Specifically, we assume the following multiplicity/degeneracy structure for the eigenvalues   
\begin{eqnarray}
 &\hskip-2mm \lambda_{\sum_{i=1}^{l-1}m_i+1}
=\lambda_{\sum_{i=1}^{l-1}m_i+2} 
= \ldots 
=\lambda_{\sum_{i=1}^{l-1}m_i+m_l}\equiv \lambda_{[l]},
\end{eqnarray}
and represent this block of the eigenvalues with the diagonal matrix $\Lambda_l=\lambda_{[l]}I_{m_l\times m_l}$, subject to $\sum_{l=1}^{d}m_l=n$ ($d$ distinct blocks). Thus,
\begin{eqnarray}
& \varrho=V (\oplus_{l}\Lambda_l) V^*,
\label{spectral-1}
\end{eqnarray}
where $V\in U(n)$ is a unitary matrix comprised of the eigenvectors of $\varrho$ arranged in its columns. The subspace $E_k$ corresponding to $\Lambda_k$ can be uniquely identified with the eigenprojection 
\begin{eqnarray*}
& \mathcal{P}_k\equiv \sum_{l=1}^{m_k}|e_{\sum_{i=1}^{k-1}m_i+l}\rangle \langle e_{\sum_{i=1}^{k-1}m_i+l}|. 
\end{eqnarray*}
Now we assume that we want to unitarily preserve the blocks $\oplus_{k\in K}\Lambda_k=:\Lambda_K$ for some given subset $K$ of the blocks, while the complementary eigenvalues $\Lambda_{\K}(t)$ are not determined but in that their multiplicities that remain unchanged. Hence Eq.~(\ref{spectral-1}) reads as
\begin{eqnarray*}
\varrho(t)=V(t)~\mathrm{diag}\left\{\Lambda_K,\Lambda_{\K}(t)\right\}~V^*(t),
\end{eqnarray*}
and $|e_{i\in I_K}(t)\rangle=\widetilde{V}_K(t)|e_{i\in I_K}(0)\rangle$ for some $\widetilde{V}_K(t)\in U(n)$ [with $\widetilde{V}_K(0)\equiv I$], where $I_K$ denotes the set $\{i\}$ of indices such that the corresponding eigenvalues $\lambda_i$ are in $\Lambda_K$.

We now define the DFS as follows:
\begin{eqnarray}
\varrho_{\text{DFS}}(t)\equiv\mathcal{P}_{\text{DFS}}(t)\varrho(t)\mathcal{P}_{\text{DFS}}(t)/\mathrm{Tr}[\mathcal{P}_{\text{DFS}}(t)\varrho(t)\mathcal{P}_{\text{DFS}}(t)],
\label{ourDFS}
\end{eqnarray}
in which $\mathcal{P}_{\text{DFS}}(t)\equiv\sum_{k\in K}\mathcal{P}_k(t)$. Note that $\mathrm{Tr}[\varrho_{\text{DFS}}(t)]=1$ and $\mathcal{P}_{\text{DFS}}(t)\varrho(t)=\varrho(t)\mathcal{P}_{\text{DFS}}(t)=\sum_{i\in I_K}\lambda_{i}|e_{i}(t)\rangle\langle e_{i}(t)|$;  hence 
\begin{eqnarray}
\varrho_{\text{DFS}}(t) &=& \textstyle{\sum_{i\in I_K}}\lambda_{i}|e_{i}(t)\rangle\langle e_{i}(t)|/\textstyle{\sum_{i\in I_K}}\lambda_{i} \nonumber\\
&=&\widetilde{V}_K(t)\varrho_{\text{DFS}}(0)\widetilde{V}_K^*(t).
\end{eqnarray}
That is, we obtain the following dynamical equation of motion for DFS:
\begin{eqnarray}
& \dot{\varrho}_{\text{DFS}}(t) = -\imath[H_{\text{DFS}}(t),\varrho_{\text{DFS}}(t)],
\label{rhodotDFS}
\end{eqnarray}
where dot denotes $\partial_t$ and 
\begin{eqnarray}
& H_{\text{DFS}}(t) =\imath \dot{\widetilde{V}}_K(t)\widetilde{V}_K^*(t),
\label{HDFS}
\end{eqnarray}
is the Hamiltonian for the corresponding sub-dynamics. 

\subsection{New interpretation of control-invariance}

Note that the unitary sub-dynamics should be compatible with the equation of motion for $\varrho(t)$ [Eq.~(\ref{e:liouville_von_neumann})]. From Eq.~(\ref{ourDFS}) we have
\begin{eqnarray*}
\hskip-2mm & \dot{\varrho}_{\text{DFS}}= (\dot{\mathcal{P}}_{\text{DFS}}\varrho\mathcal{P}_{\text{DFS}}+ \mathcal{P}_{\text{DFS}} \dot{\varrho} \mathcal{P}_{\text{DFS}}+ \mathcal{P}_{\text{DFS}}\varrho\dot{\mathcal{P}}_{\text{DFS}})/\sum_{i\in I_K}\lambda_i,
\end{eqnarray*}
where we have used the fact that $\sum_{i\in I_K}\lambda_i$ is constant. From the definition of $\mathcal{P}_{\text{DFS}}(t)$ it is evident that 
\begin{eqnarray*}
 & \dot{\mathcal{P}}_{\text{DFS}}(t)=-\imath[H_{\text{DFS}}(t),\mathcal{P}_{\text{DFS}}(t)].
\end{eqnarray*}
After some algebra and taking Eq.~(\ref{rhodotDFS}) into account, we obtain
\begin{eqnarray}
\hskip-4mm& \mathcal{P}_{\text{DFS}}(t) \dot{\varrho}(t) \mathcal{P}_{\text{DFS}}(t) =-\imath \mathcal{P}_{\text{DFS}}(t) [H_{\text{DFS}}(t),\varrho(t)] \mathcal{P}_{\text{DFS}}(t).
\label{consistency}
\end{eqnarray}
One can simplify the above equation further and extract a relation including $\widetilde{V}_K(t)$, $H_{\text{DFS}}(t)$, $H_0$, $\{H_u\}$, $\{F_{\alpha}\}$, and $\varrho(0)$. However, because of the existence of the terms $\mathcal{P}_{\text{DFS}}F_{\alpha}\varrho F^*_{\alpha}\mathcal{P}_{\text{DFS}}$ [coming from $\mathcal{P}_{\text{DFS}}L(\varrho)\mathcal{P}_{\text{DFS}}$] we do need the whole $\varrho(t)$ in order to characterize DFS. A sufficient condition for bypassing this need is to seek for DFS among the common eigenvectors of $\{F_{\alpha}\}$. That is, consider vectors $|\psi_l\rangle\in\mathcal{H}$ such that $F_{\alpha}|\psi_l\rangle=c_{\alpha}|\psi_{l}\rangle$~$\forall \alpha$, and confine our search for DFS in the subspace spanned by $\{|\psi_l\rangle\}$. This is reminiscent of the traditional DFS condition \cite{x1,DFS_from_Alireza}.

In summary, given a desired DFS Hamiltonian $H_{\text{DFS}}(t)$ and $K$, and specifying a model comprising of the Hamiltonian $H_0$, a Lindbladian $L(\cdot)$ [i.e., the set $\{(\gamma_{\alpha},F_{\alpha})\}$], a control Hamiltonian set $\{H_{\alpha}\}$, and $\varrho(0)=\varrho_0$ [hence $\mathcal{P}_{\text{DFS}}(0)$], one can in principle solve Eq.~(\ref{consistency}) to find appropriate control knobs $u_{\alpha}(t)$. Clearly, the solution does not need to be unique or always exist.

\subsection{Evolution of DFS on Grassmannian manifold}

As emphasized in the Subsection~\ref{s:fconcepts}, 
here, the crucial mathematical object of concern is the unitary evolution 
$\widetilde{V}_K(t)$ of the eigenvectors associated with the constant 
eigenvalues $\Lambda_K$ of the density operator. To put it simply, 
if we choose the canonical basis 
$\left(0, \ldots, 0,1,0,\ldots,0\right)^T$ for $\left| e_{i \in I_K}(0)\right\rangle$, 
$\widetilde{V}_K(t)$ can be viewed as the $U(n)$-matrix partitioned as 
$\left(\widetilde{V}_{K,K}(t) ~~ \widetilde{V}_{K,\K}(t)\right)$, where $\widetilde{V}_{K,K}(t)$ 
denotes the matrix made up with the columns $\left| e_{i\in I_K}(t)\right\rangle$. 
However, the only specification on $\widetilde{V}_{K}(t)$ is that it should map 
the orthonormal $m_K$-frame [$m_K\equiv \sum_{k\in K}m_k$] $\left|e_{i\in I_K}(0)\right\rangle$ 
to the orthonormal $m_K$-frame $\left|e_{i\in I_K}(t)\right\rangle$, 
regardless of the remaining $(n-m_{K})$-frame in the orthogonal complement. 
Thus, $\widetilde{V}_{K}(t) \in U(n)/U(n-m_K)$. 
The latter is the Stiefel manifold $\mathbb{V}_{m_K}(\mathbb{C}^n)$ of $m_K$-frames in $\mathbb{C}^n$.  

With the preceding concepts, $\mathrm{DFS}(t)$ is the column span of $\widetilde{V}_{K,K}(t)$, 
as as such $\mathrm{DFS}(t) \in U(n)/U(m_K)\times U(n-m_K)$. 
The latter is the Grassmannian manifold $\mathbb{G}_{m_K}(\mathbb{C}^n)$ 
of $m_K$-dimensional complex subspaces of $\mathbb{C}^n$. 

All of the above concepts are intertwined in the following fiber map, 
the principal bundle of the well-known universal bundle with $U(m_K)$-structure group:
\begin{small}
\begin{eqnarray*}
 & \begin{array}{ccccc}
U(m_k) & \stackrel{i}{\longrightarrow} & U(n)/U(n-m_K)& \ni  & \widetilde{V}_K(t)\\
       &             & \downarrow \pi                                   &      & \downarrow \\
       &             & U(n)/U(m_K)\times U(n-m_K)                   &  \ni & \mathrm{DFS}(t)
\end{array},
\end{eqnarray*}
\end{small}
In the above, $i$ is the inclusion and $\pi$ the bundle projection~\cite[Sec. 25.7-8]{Steenrod1951},
\cite[Chap. 8, Theorem 3.6 and Corollary 3.7]{Husemoller1994}. 

As is well-known, this bundle is far from trivial. Accordingly, it need not, and will not in general, have a cross-section. 
The consequence is that we might not have a globally defined $\widetilde{V}_{K}(t)$. 
Another corollary is that we might not have 
a globally defined, at least continuous, orthonormal basis in $E_K(\varrho)$.  Since the above is a principal bundle, 
existence of a cross section is equivalent to whether the bundle is (globally) trivial~\cite[Chap. 4, Corollary 8.3]{Husemoller1994}. 
This is obviously not the case. 
For example, if $\varrho_{\text{DFS}}$ is $2 \times 2$ and $\Lambda = \lambda \in \mathbb{R}^+\cup\{0\}$, 
the above fibration reduces to the Hopf fibration
\begin{eqnarray*}
& \begin{array}{ccc}
S^1 & \longrightarrow & S^3 \\
       &             & \downarrow \\
       &             & S^2 
\end{array}. 
\end{eqnarray*}
To see this, observe that $U(2)/U(1)\cong S^3$ (see~\cite[Sec. 7.10]{Steenrod1951}) 
and that $U(2)/U(1) \times U(1) \cong \mathbb{C}\mathbb{P}^1 \cong S^2$, 
where $\mathbb{C}\mathbb{P}^1$ denotes the complex projective line (see~\cite[Sec. 20.1]{Steenrod1951}). 
The latter is a prototypical bundle that has no cross section.

\section{Decoherence-free vector bundle}
\label{s:differential_geometry}

\subsection{Set-up}

First, we fix some notations. Let $\Herm(n)$ be the set of $n \times n$ Hermitian matrices. The Liouville-von Neumann system evolves over the convex set of positive definite Hermitian matrices of trace $1$. Let $\D(n) \subset \Herm(n)$ be this space of density matrices. We will drop the argument $n$ when there is no danger of confusion. The spectrum of any $n \times n$ density operator can be written as 
\begin{eqnarray*}
\lambda_1 &=& \lambda_2 = \ldots = \lambda_{m_1} >\lambda_{m_1+1} = \lambda_{m_1+2} = \ldots = \lambda_{m_1+m_2} \\
&>& \ldots > \lambda_{\sum_{i=1}^{d-1}m_i+1} =  \ldots = 
\lambda_{\sum_{i=1}^{d-1}m_i+m_d}.
\end{eqnarray*}

\subsection{The Lie group perspective \cite{altafini}}

Define $\mu=(m_1,m_2, \ldots, m_d)$, and let $\Herm_\mu$ be the stratum of {\it Hermitian matrices} having this multiplicity structure. In~\cite[Th. 4.11]{GutkinJonckheereKarow}, it was shown that $\Herm_\mu$ is a real-analytic $\mathbb{R}^*$-homogeneous submanifold of codimension $\sum_{i=1}^d m_i^2 -d$ in $\Herm$. Now, consider the stratum $\D_\mu$ of the density matrices having this multiplicity structure. 

\begin{lemma}
$\D_\mu$ is a real-analytic submanifold 
of codimension $\sum_{i=1}^d m_i^2 - d +1$ in $\Herm$.
\end{lemma}

\begin{proof}
This is corollary of the real-analytic manifold property of $\Herm_\mu$, 
along with the analytic implicit function theorem~\cite[Th. 2.5.3]{primer_analytic}. 
\end{proof}

Let $\M=\{(m_1,m_2,\ldots,m_d): \sum_{i=1}^d m_i=n, 2 \leq d \leq n\}$. 
Clearly, $\sqcup_{\mu \in \M} \D_\mu$ is a {\it stratification}~\cite{GoreskyMacPherson1988} of $\D$.
Furthermore, each stratum $\D_\mu$ is foliated~\cite{Tamura} by leaves $\D_{\mu,\Lambda}$, 
where $\Lambda=\oplus_{i=1}^{d}\Lambda_i$.  
As is well known (see e.g., \cite{altafini}), $\D_{\mu,\Lambda}(n) \cong U(n)/\prod_{l=1}^d U(m_l)$. Clearly, this foliation of $\D_\mu$ has codimension $n-d$. 

If $\gamma =0$, the evolution is unitary, and the density operator remains in the same stratum~\cite{altafini}. 
With $\gamma \neq0$, the evolution is nonunitary, moving from one stratum to another. The problem is to find out a control $u(t)$ such that, subject to $\gamma$, the evolution is ``partially unitary,''  in the sense that at least some but not all eigenvalues are preserved. 

In the sequel, rather than proceeding from a Lie group perspective, we work in the category of real-analytic manifolds and real-analytic maps. 

\subsection{Eigenspace vector bundle}
\label{s:the_problem}

Here, we specifically develop the controlled-invariance approach to the Liouville-von Neumann equation~(\ref{e:liouville_von_neumann}). This approach consists in defining a DFS to be an eigenspace of $\varrho(t)$ along which the evolution is unitary. As it has been shown in Sec.~\ref{s:newer}, the latter is equivalent to some blocks of eigenvalues of $\varrho(t)$, $\Lambda_K$, being preserved along the motion. Let $\D_{\Lambda_K} \subseteq \D$ be the subset of density operators across which the set of eigenvalues contains the blocks $\Lambda_{k\in K}$.  
The system evolves over that space and the DFS, $\oplus_{k\in K}E_k(\varrho)$, is a vector space ``above'' $\varrho \in \D_{\Lambda_K}$. Collect all such eigenspaces in the disjoint union $\E= \sqcup_{\varrho \in \D_{\Lambda_K}} E_k(\varrho)$, 
topologized as a subspace of $\D_{\Lambda_K} \times \mathbb{C}^{\sum_{k \in K} m_k}$. This obviously leads to the complex vector bundle formulation: 
\begin{eqnarray*}
& \begin{array}{ccl}
\mathbb{C}^{\sum_{k \in K}m_k} & \stackrel{i}{\rightarrow} & \E \\
                               &             &  \downarrow \pi \\
                               &             &  \D_{\Lambda_K}
\end{array}. 
\end{eqnarray*}
The collection of all eigenspaces is the total space $\E$, 
$\pi$ is the projection $\E \ni (\varrho,e) \mapsto \varrho \in \D_{\Lambda_K}$, 
and $\mathbb{C}^{\sum_{k\in K} m_k} \cong \pi^{-1}(\varrho)$ is the fiber. 
This bundle formulation clarifies the difference between the DF {\it manifold}, $\D_{\Lambda_K}$, 
and the DF {\it subspace}, $\oplus_{k \in K} E_k$. 

For $\varrho(0) \in \D_{\Lambda_K}$, it is required that the evolution remains in $\D_{\Lambda_K}$, viz., $\varrho(t) \in \D_{\Lambda_K}$, with $\D_{\Lambda_K}$ as large as possible, and the decoherence process is, in some local coordinate patch, confined to $\D_{\Lambda_K}^\perp$. 

\subsection{Real-analytic DFM}

We now proceed to the topology of the base space $\D_{\Lambda_K}$ of the fiber bundle identified in Section~\ref{s:the_problem}. We first somewhat restrict the base space, in the sense that not only do we specify the blocks $\Lambda_K$  to be preserved, but in addition, we require the multiplicity structure to remain constant in the complementary blocks $\bar{K}$. This space is denoted as $\D_{\Lambda_K,m_{\bar{K}}}$. 

\begin{theorem}
\label{t:real_analytic}
$\D_{\Lambda_K,m_{\bar{K}}}$ is a real-analytic manifold of real dimension 
$n^2+\sum_{\bar{k}\in \bar{K}} m_{\bar{k}} - \sum_{\bar{k}} m_{\bar{k}}^2 -\sum_{k} m_k^2-1$.
\end{theorem}
\begin{proof}
The proof follows by an adaptation of the proof of~\cite[Th. 4.11, Appendix B]{GutkinJonckheereKarow}. 
We first temporarily disregard the positive definiteness and trace properties and prove the real-analyticity of $\Herm_{\Lambda_K,m_{\bar{K}}}$ [$m_{\bar{K}}\equiv\sum_{\bar{k}\in \bar{K}} m_{\bar{k}}$] and compute its real dimension; 
then we will introduce the constraints specific to a density operator. 

Consider the mapping
\begin{eqnarray*}
\eta: U(n) \times {D}_{\Lambda_K,m_{\bar{K}}} & \rightarrow & \Herm_{\Lambda_K,m_{\bar{K}}} \\
(V,D) & \mapsto & VDV^*
\end{eqnarray*}
where $D_{\Lambda_K,m_{\bar{K}}}$ is the set of block-diagonal matrices with the blocks $\Lambda_{k \in K}$ specified, while the only thing specified for the complementary blocks $\bar{K}$ is their multiplicity structure. The obvious surjective property of this mapping will be used to construct the coordinate chart. The problem is that the mapping is many-to-one. Hence in order to determine $\dim_\mathbb{R}[\Herm_{\Lambda_K,m_{\bar{K}}}]$ and construct a coordinate chart, we need to determine $\dim_\mathbb{R}[U(n) \times {D}_{\Lambda_K,m_{\bar{K}}}]$ and the dimension of the kernel of $d\eta$. It is well known that $\dim_{\mathbb{R}} [U(n)]=n^2$. On the other hand, observe that $\dim_\mathbb{R}[{D}_{\Lambda_K,m_{\bar{K}}}]=m_{\bar{K}}$. Thus, $\dim_{\mathbb{R}}[U(n) \times {D}_{\Lambda_K,m_{\bar{K}}}]=n^2+m_{\bar{K}}=:p$. Next, we need to evaluate the dimension of the kernel of the differential 
\begin{eqnarray*}
\hskip-4mm & d_{(V_0,D_0)} \eta : T_{V_0}U(n) \times T_{D_0} {D}_{\Lambda_K,m_{\bar{K}}} 
  \rightarrow T_{(V_0,D_0)}\Herm_{\Lambda_K,m_{\bar{K}}},
\end{eqnarray*}
wherein $T_{V_0}U(n)=\left\{\imath HV: H \in \Herm(n) \right\}$ [$H$ is defined by $V=e^{\imath H}$], and 
\begin{eqnarray}
&\hskip-3mm T_{D_0} {D}_{\Lambda_K,m_{\bar{K}}} =\left\{\left(\begin{smallmatrix}
\mathbf{0}&\mathbf{0}  \\
\mathbf{0} & \oplus_{\bar{k}\in \bar{K}}\delta_{\bar{k}} I_{m_{\bar{k}}\times m_{\bar{k}}}
\end{smallmatrix}\right): \delta_{\bar{k}} \in \mathbb{R}\right\}. 
\label{e:conformably}
\end{eqnarray}
Take $(\Delta_1,\Delta_2) \in T_{V_0}U(n) \times T_{D_0} {D}_{\Lambda_K,m_{\bar{K}}}$. After some algebra, we find that 
\begin{eqnarray*}
d_{(V_0,D_0)} \eta(\Delta_1,\Delta_2) &=& \imath H\eta(\Delta_1,\Delta_2) - \imath\eta(\Delta_1,\Delta_2) H\nonumber\\
&& +V\Delta_2 V^* 
\end{eqnarray*}
If we partition $V$ conformably with Eq.~(\ref{e:conformably}), viz., 
\begin{eqnarray*}
 & V=\left(\begin{array}{cc}
V_{KK} & V_{K\bar{K}} \\
V_{\bar{K}K} & V_{\bar{K}\bar{K}}
\end{array}\right),
\end{eqnarray*}
we see
\begin{eqnarray*}
 V\Delta_2 V^* =
\left(\begin{array}{c}
V_{K\bar{K}} \\
V_{\bar{K}\bar{K}}
\end{array}\right)
\oplus_{\bar{k}\in \bar{K}}\delta_k I_{m_{\bar{k}}\times m_{\bar{k}}}
\left(\begin{array}{cc}
V_{K\bar{K}}^* & V_{\bar{K}\bar{K}}^*
\end{array}\right),
\end{eqnarray*}
so that the kernel of the mapping $\Delta_2 \mapsto V\Delta_2 V^*$ is $\{0\}$. On the other hand, the set of $H$ such that $H\eta=\eta H$ is the set of those $H$'s having the same (normalized) eigenvectors as $\eta$. But the normalized eigenvectors associated with an $m \times m$ block of eigenvalues are defined up to an element of $U(m)$. Therefore, the real dimension of the kernel of $H \mapsto H\eta-\eta H$ is $\sum_{k} m_k^2 + \sum_{\bar{k}} m_{\bar{k}}^2$, whence $\dim_\mathbb{R}[\ker (d\eta)]=\sum_{k} m_k^2 + \sum_{\bar{k}} m_{\bar{k}}^2$ and $\dim_\mathbb{R}[\Herm_{\Lambda_K,m_{\bar{K}}}]=n^2+m_{\bar{K}}-\sum_{k} m_k^2 - \sum_{\bar{k}} m_{\bar{k}}^2=:q$. 

Since $U(n)\times D_{\Lambda_K,m_{\bar{K}}} \cong \mathbb{R}^p$, $\eta$ can be viewed as a mapping from $\mathbb{R}^p$ {\it into} $\mathbb{R}^{n^2} \cong \Herm(n)$. Since in this trivial parameterization of the domain and the image matrices by their entries, 
the mapping $\eta$ is merely matrix multiplication, it is obviously analytic. Furthermore, since the rank of $d\eta$ is constant, it follows from the analytic version of the constant rank theorem~\cite[Remark 2.5.4]{primer_analytic}, \cite[Rank Theorem]{PMIHES_constant_rank} 
that there exist neighborhoods $N_1$ of $(V,D)$ and $N_2$ of $\eta(V,D)$ with $\eta(N_1) \subset N_2$ and analytic maps $\phi: N_1 \rightarrow \mathbb{R}^p$ and $\gamma: N_2 \rightarrow \mathbb{R}^{n^2}$, as shown by the following diagram:
\begin{eqnarray*}
 & \begin{smallmatrix}
\mathbb{R}^p \cong U(n) \times {D}_{\Lambda_K,m_{\bar{K}}} \supset & N_1  & \stackrel{\eta}{\longrightarrow} & N_2 \subseteq \D_{\Lambda_K,m_{\bar{K}}} \subset \mathbb{R}^{n^2} \cong \Herm(n)& \\
& \phi \downarrow &                                                       & \downarrow \gamma & \\
  & \mathbb{R}^p & \stackrel{\gamma \circ \eta \circ \phi^{-1}}{\longrightarrow} & \mathbb{R}^{n^2} & 
\end{smallmatrix} 
\end{eqnarray*}
and such that
\begin{eqnarray}
\gamma \circ \eta \circ \phi^{-1}(x_1,\ldots,x_p)=(x_1,\ldots,x_q,0,0,\ldots,0).
\label{e:submanifold}
\end{eqnarray} 
It follows from Eq.~(\ref{e:submanifold}) that in $\mathbb{R}^{n^2}$ charted by coordinates $\left( x_1,\ldots,x_q,x_{q+1},\ldots,x_{n^2}\right)$, $N_2 \supset \eta(N_1)$ is characterized by $\left(x_1,\ldots,x_q,0,\ldots,0\right)$. Therefore, $\Herm_{\Lambda_K,m_{\bar{K}}}$ is a real-analytic submanifold of $\mathbb{R}^{n^2}$ and that the restriction $\gamma\vert \eta(N_1)$ is a coordinate chart. 

Finally, we impose the trace and positive definiteness conditions. Regarding the trace condition, we have to look more carefully at the way the local coordinates of the various spaces could be set up. In $U(n)\times D_{\Lambda_K,m_{\bar{K}}}$, 
the eigenvalues $(\lambda_1, \lambda_2,\ldots,\lambda_n)$ of $D$ are among the obvious local coordinates. These coordinates can be changed by an affine nonsingular transformation to $(\mathrm{Tr}[D]-1,\lambda_2,\lambda_3,\ldots,\lambda_n)$; it follows that we can set $x_1=\mathrm{Tr}[D]-1$. On the other hand, in $\Herm(n)$, the obvious local coordinates are the entries of the matrix $VDV^*$. Again, an affine nonsingular change of coordinates allows us to take 
$(\sum_{i=1}^n [VDV^*]_{ii}-1,\{[VDV^*]_{ii}:i\neq1\};\{[VDV^*]_{ij}: i \neq j\})$ as local coordinates, hence consistent with $x_1=\mathrm{Tr}[VDV^*]-1$. It follows that the subset of $\D_{\Lambda_{K},m_{\bar{K}}}$ subject to the trace 1 condition is characterized by $(0,x_2,\ldots,x_q,0,0,\ldots,0)$ and is hence an analytic submanifold of $\Herm(n)$ 
(see~\cite[II.2]{lang}).  Hence choose $(x_2,\ldots, x_q)$ as local coordinates for this submanifold. The positive definiteness condition is just a matter of restricting the $(x_2,\ldots,x_q)$ coordinates so that $\gamma^{-1}(0,x_2,\ldots,x_q,0,0,\ldots,0)$ is positive definite. $\D_{\Lambda_K,m_{\bar{K}}}$ is a real-analytic manifold of real dimension $q-1=n^2+m_{\bar{K}}-\sum_{k} m_k^2 - \sum_{\bar{k}} m_{\bar{k}}^2-1$. 
\end{proof}

Care must be taken when applying this dimension result to the extreme case of $n-1$ eigenvalues specified, as the density property immediately implies that the $n$th eigenvalue is also specified. As such, {\it all} eigenvalues are specified and hence the dimension result of the theorem reduces to a trivial corollary of~\cite[Th. 4.11]{GutkinJonckheereKarow}. As we illustrate below, the above theorem has to be applied with $m_{\bar{k}}=1$, as the dimension formula takes the density property under  consideration. Regarding $\{m_{k\in K}\}$, we consider two extreme cases: when the $(n-1)$ formally specified eigenvalues are (i) all equal and (ii) pairwise distinct: (i) If all explicitly specified eigenvalues  are all equal, $m_1=m_2=\ldots =m_{n-1}=1$, we obtain, $n^2+m_K - \sum_{k=1}^{n-1} m_k^2-\sum_{\bar{k}=n} m_{\bar{k}}^2-1=n^2+1-(n-1)^2-1-1=2n-2$. From~\cite[Th. 4.11]{GutkinJonckheereKarow}, it follows that the set of all $n\times n$ Hermitian matrices with an $(n-1)$ Jordan block and a remaining eigenvalue of multiplicity $1$ is $n^2-((n-1)^2 + 1 - 2)=2n$. But in this theorem, the numerical values of the two distinct eigenvalues are not specified, so that if they become specified, the dimension drops to $2n-2$, consistently with the above result. (ii) If all explicitly specified eigenvalues are pairwise distinct, $m_1=n-1$, $m_2=1$, and $n^2+m_K - \sum_{k=1}^{1} m_k^2-\sum_{\bar{k}=2}^2 m_{\bar{k}}^2-1=n^2+1-((n-1)+1)-1=n^2-n$. With all eigenvalues free,~\cite[Th. 4.11]{GutkinJonckheereKarow} yields $n^2-(n-n)=n^2$. With all $n$ eigenvalues specified, the dimension drops to $n^2-n$, consistently with our result. 

From Theorem~\ref{t:real_analytic} we can compute the real dimension of the various $\D_{\Lambda_K,m_{\bar{K}}}$. Table~\ref{t:dimD} shows some results for $n=4$.

\begin{table}[bp]
\caption{Dimension of various DFMs in the $n=4$ (e.g., 2-qubit) case.}
\begin{center}
\begin{tabular}{|c|c|c|}\hline
$\{m_k:k\in K$\} & $\{m_{\bar{k}}:\bar{k}\in\bar{K}\}$ & $\dim [\D_{\Lambda_K,m_{\bar{K}}}]$ \\ \hline \hline 
\{1\}       & \{1,1,1\}   & 14  \\ \cline{2-3}
            &\{1,2\}      & 12  \\ \cline{2-3}
            & \{3\}       & 8   \\ \hline
\{1,1\}     & \{1,1\}     & 13  \\ \cline{2-3}
            & \{2\}       & 11  \\ \hline
\{2\}       & \{1,1\}     & 11  \\ \cline{2-3}
            & \{2\}       & 9   \\ \hline
\{1,1,1,1\} & $\emptyset$ & 12  \\ \cline{2-3}
\{1,1,2\}   & $\emptyset$ & 10  \\ \cline{2-3}
\{2,2\}     & $\emptyset$ & 8   \\ \cline{2-3}
\{1,3\}     & $\emptyset$ & 6   \\ \cline{2-3}
\{4\}       & $\emptyset$ & 0   \\ \hline
\end{tabular}
\end{center}
\label{t:dimD}
\end{table}

\section{Reachability distribution}

The overall objective is to confine the density matrix within a DFM. As a preliminary step in narrowing down as to whether this can be achieved, 
we introduce a \textit{reachability distribution} $\V$, which foliates $\D$ with integral manifolds (of the distribution) over which the $\varrho$-trajectories are confined, given an initial condition $\varrho(0)$, given a decoherence rate process $\gamma$ (either constant or stochastically varying), given an arbitrary control $u(t)$.  Precisely, if $\Gamma$ is the set of allowable decoherence rates, one such leaf $\L_{\varrho_0}$ contains all $\varrho(t)$, $t \geq 0$, such that $\varrho(0)=\varrho_0$ and $\varrho(t)$ is solution to~(\ref{e:liouville_von_neumann}) for some $\gamma \in \Gamma$ and some $u \in C^0$.   
Then the bigger geometric picture is how $\V$ and $\text{DFM}$ are intertwined. 
A sufficient condition for decoherence immunity is $T(\text{DFM}) \supseteq \V$.  
However, since we have preserved full control authority in $\V$, 
should $T(\text{DFM}) \cap \V \not = \emptyset$, then it suffices to utilize that control authority to make 
$T(\text{DFM}) \cap \V$ invariant. The next step in finding a DFS would be to utilize the condition (\ref{consistency}) 
in order to find appropriate control. We shall discuss this elsewhere. 

\subsection{Bilinear models over the Bloch sphere}
\label{s:Krener}

We follow~\cite{LvNLBloch} and we define a state vector $\mathbf{x} \in \mathbb{R}^{n^2-1}$ whose components (in the standard computational basis) are $\varrho_{11}-\varrho_{ii}$, $i=2,\ldots,n$; $\varrho_{ij}+\varrho_{ji}$, $i>j$; $\imath(\varrho_{ij}-\varrho_{ji})$, $i>j$. Then Eq.~(\ref{e:liouville_von_neumann}) can be rewritten in bilinear format:
\begin{eqnarray}
& \dot{\mathbf{x}}(t) = A\mathbf{x}+\sum_\alpha (B_\alpha \mathbf{x}) u_\alpha+\sum_\alpha (G_\alpha \mathbf{x}) \gamma_\alpha(t).
\label{e:gbilinear}
\end{eqnarray}
The matrices $A$, $B$, and $G$ are obtained as in~\cite{LvNLBloch}. 

In the stochastically-varying decoherence rate model, 
the reachability distribution $\V$ is defined as the smallest distribution such that 
\begin{eqnarray}
&& \left[ A\mathbf{x},\V \right] \subseteq \V + B\mathbf{x}, \quad 
\left[ B\mathbf{x},\V \right] \subseteq \V + B\mathbf{x}, \label{e:localinvariance}\\
&& \V \supseteq G\mathbf{x} \label{e:localinvariance-2}
\end{eqnarray}
On the other hand, if the decoherence rates are constant, $\V$ is defined as the smallest distribution such that
\begin{eqnarray*}
& \left[ (A+\sum_\alpha G_\alpha \gamma_\alpha)\mathbf{x},\V \right] \subseteq \V + B\mathbf{x}, \quad 
\left[ B\mathbf{x},\V \right] \subseteq \V + B\mathbf{x}. 
\end{eqnarray*}
The condition that $\varrho(0)$ should be in the integral manifold of $\V$ makes it nontrivial.  

\subsection{Algorithm}

Since we are working in the category of real-analytic functions, it is natural to look for a distribution $\V$ whose basis $\{v_0(\mathbf{x}),v_1(\mathbf{x}),\ldots\}$ is at least locally analytic in $\mathbf{x}$. Plugging $v_i(\mathbf{x})$ in Eq.~(\ref{e:localinvariance}), it is easily seen that one can work out the equations one degree at a time. In particular, the only degree-0 term that satisfies~(\ref{e:localinvariance}) is a common real eigenvector of $A$ and $B$. But no such eigenvector exists; 
hence the controlled-invariant distribution, if any, cannot have a constant term. The degree-1 terms are of the form $V_i \mathbf{x}$, where $V_i$ is an $n \times n$ real matrix. The terms of degree-2, e.g., $v^{(2)}$, and higher must be invariant under both $A\mathbf{x}$ and $B\mathbf{x}$: $[A\mathbf{x},v^{(2)}(\mathbf{x})]\subseteq v^{(2)}(\mathbf{x})$,  $[B\mathbf{x},v^{(2)}(\mathbf{\mathbf{x}})]\subseteq v^{(2)}(\mathbf{x})$. Moreover, $v^{(2)}(\mathbf{x}) \subseteq G\mathbf{x}$. 

Let us agree to stop the expansion at degree-1. Thus we are looking for a $(A\mathbf{x},B\mathbf{x})$ controlled-invariant distribution of the form $\text{span}\{V_0 \mathbf{x},V_1\mathbf{x},V_2 \mathbf{x},\ldots\}$ containing $G\mathbf{x}$. 

The essential idea is to start from $V_0 \mathbf{x}=G\mathbf{x}$ and define $V_1$ such that $[A\mathbf{x},G\mathbf{x}] \subseteq V_1\mathbf{x}+B\mathbf{x}$; thereafter, we iterate Eq.~(\ref{e:localinvariance}) until the distribution stabilizes. We will start from $[A,G]\subseteq V_1+B$. In doing so one must bear in mind that, should $V_i$ and $V_j$ be linearly dependent over $\mathbb{R}$, so are $V_i \mathbf{x}$ and $V_j \mathbf{x}$. Therefore, we will seek matrices $V_i$ and $V_j$ linearly independent over $\mathbb{R}$.  

\section{Two-qubit example}

Consider a system of two qubits ($n=4$) each of which independently going under pure dephasing process with no internal Hamiltonian and full control,  
\begin{eqnarray*}
& \dot{\varrho}=-\imath[H_uu,\varrho]+\sum_{\alpha=1}^{2}(Z_{\alpha}\varrho Z_{\alpha}-\varrho)\gamma_\alpha,
\end{eqnarray*}
where $H_{u}u= \frac{1}{2}\sum_{\alpha=1}^{2}X_{\alpha}u^{x}_{\alpha}+Y_{\alpha}u^{y}_{\alpha}+Z_{\alpha}u^{z}_{\alpha}$. The operator $F_\alpha\equiv Z_{\alpha}$ acts nontrivially on the $\alpha$th qubit; e.g., $Z_1=\sigma_z \otimes I, \quad Z_2=I \otimes \sigma_z$, and similarly for $X_{\alpha}$ and $Y_{\alpha}$ related to the Pauli operators $\sigma_{x}$ and $\sigma_{y}$. Note that the absence of free dynamics ($H_0=0$) does not make the invariance issue trivial. Indeed, in this bilinear case, Eq.~(\ref{e:localinvariance-2}) is certainly nontrivial. However, because of the absence of free dynamics, we do not have to run an iteration on Eq.~(\ref{e:localinvariance}).  

Observe that, here, for convenience in the writing of the code, 
we have kept $\mathrm{Tr}[\varrho$] in its vector representation. This yields a real, $16$-dimensional model of the form
\begin{eqnarray*}
& \dot{\mathbf{x}}=\sum_{\alpha=1}^2 \sum_{\omega=x}^z \left(B^\alpha_\omega \mathbf{x} \right) u^\alpha_\omega 
+\sum_{\alpha=1}^2 \left(G^\alpha \mathbf{x}\right) \gamma^\alpha .
\end{eqnarray*}
We have $[B^i_\alpha,B^j_\beta]=B^i_\gamma\delta_{ij}$, for $(\alpha,\beta,\gamma)$ a cyclic permutation of $(x,y,z)$ and for $i=1,2$. The $B_x^i,B_y^i$ matrices are \textit{not} skew-symmetric, while the $B_z^i$ matrices are skew-symmetric. The $G$ matrices, however, are diagonal, hence symmetric. 

In this 2-qubit case, the recursive algorithm that finds the solution to Eqs.~(\ref{e:localinvariance}) and (\ref{e:localinvariance-2}) 
stabilizes before the full 16-dimensional space is reached. 
Specifically, $\dim[\mathcal{V}]=10$. 
This already bodes well on existence of DFSs. 
We also looked at the $\gamma^1=\gamma^2$ case; it suffices to set $G:=G^1+G^2$ in the preceding and compute $\V$ by iterating 
Eq.~(\ref{e:localinvariance}) starting from $G$. Then it can be verified that $\dim[\V]=9$. Since the decoherence process is 
constrained compared with the preceding case, a smaller-dimensional $\V$ was indeed to be expected. 

From Theorem~\ref{t:real_analytic} and Table~\ref{t:dimD}, it is evident that there are several 10-dimensional manifolds that can be embedded in various DFMs retaining $1$ or $2$ eigenvalues. 
Thus if it can be proved that, in addition to the dimensional matching, we have $T(\text{DFM}) \supseteq \V$, 
we will have some guaranteed DFSs. 


\section{Conclusion}

We have developed a control-assisted generalized quantum decoherence-free concept. The condition the driving control fields must satisfy to lead to a desired decoherence-free dynamics has been given [Eq.~(\ref{consistency})]. Presuming the existence of such decoherence-free sub-dynamics, we have identified a corresponding decoherence-free manifold in which the dynamics is confined. Some topological and algebraic properties of this decoherence-free manifold have been discussed. More specifically, we have shown that one can endow a complex vector bundle structure to this manifold such that the fiber is the decoherence-free subspace. 

\textit{Acknowledgments.---} Helpful discussions with D. A. Lidar and S. Schirmer are acknowledged. 


\end{document}